\documentclass[conference]{IEEEtran}
\usepackage[table]{xcolor}
\usepackage{times}
\usepackage{epsfig}
\usepackage{amsmath}
\usepackage{amsthm}
\usepackage{amsfonts}
\usepackage{graphicx}
\usepackage{amssymb}
\usepackage{amstext}
\usepackage{latexsym}
\usepackage{color,colortbl}
\usepackage{ifthen}
\usepackage{multirow}
\usepackage{verbatim}
\usepackage{array,tabularx}
\usepackage{arydshln}
\usepackage[mathscr]{euscript}
\usepackage{accents}
\usepackage{cite}
\usepackage{hhline}
\usepackage{caption}
\usepackage{subcaption}
\usepackage{enumerate}
\usepackage{xcolor}
\usepackage{mathtools}
\usepackage{url}
\usepackage{xparse}
\usepackage{makecell}
\usepackage{varwidth}
\usepackage{bm}
\usepackage{arydshln}

\usepackage{comment}
\usepackage{wrapfig}

\usepackage{caption}
\usepackage{float}

\newcolumntype{C}[1]{>{\centering\let\newline\\\arraybackslash\hspace{0pt}}m{#1}}

\usepackage[normalem]{ulem}

\usepackage{amsmath,pgfplots,amssymb}
\usepackage{subcaption,graphicx}

\newcommand{\be}{\begin{equation}}
\newcommand{\ee}{\end{equation}}
\newcommand{\bear}{\begin{eqnarray}}
\newcommand{\eear}{\end{eqnarray}}
\newcommand{\bears}{\begin{eqnarray*}}
\newcommand{\eears}{\end{eqnarray*}}
\newcommand{\bi}{\begin{itemize}}
\newcommand{\ei}{\end{itemize}}
\newcommand{\ben}{\begin{enumerate}}
\newcommand{\een}{\end{enumerate}}

\DeclareMathOperator{\lcm}{lcm}
\DeclareMathOperator{\lift}{lift}
\DeclareMathOperator{\rows}{rows}

\usepackage{tikz}
\usetikzlibrary{shapes.geometric}
\usetikzlibrary{shapes,snakes,patterns}
\usetikzlibrary{arrows,positioning,automata,calc}
\usetikzlibrary{plotmarks}
\usetikzlibrary{matrix,decorations.pathreplacing}

\definecolor{ogreen}{rgb}{0,0.5,0}
\definecolor{magenta}{rgb}{1.0, 0.11, 0.81}
\definecolor{mulberry}{rgb}{0.77, 0.29, 0.55}
\definecolor{xgray}{rgb}{0.5, 0.5, 0.5}
\definecolor{ao}{rgb}{0.0, 0.5, 0.0}
\definecolor{amber}{rgb}{1.0, 0.75, 0.0}
\definecolor{capri}{rgb}{0.0, 0.75, 1.0}
\definecolor{chocolate}{rgb}{0.91, 0.41, 0.17}

\definecolor{aquamarine}{rgb}{0.87, 0.91, 0.84}
\definecolor{text}{rgb}{0.2, 0.50, 0.1}

\newtheorem{theorem}{Theorem}
\newtheorem{lemma}{Lemma}

\newtheorem{proposition}{Proposition}

\theoremstyle{definition}
\newtheorem{example}{Example}
\newtheorem{definition}{Definition}

\newtheorem{remark}{Remark}

\newcommand*{\rowstyle}[1]{% sets the style of the next row
  \gdef\@rowstyle{#1}%
  \leavevmode\@rowstyle
  \ignorespaces
}

\newcolumntype{=}{% resets the row style
  >{\gdef\@rowstyle{}\ignorespaces}%
}
\newcolumntype{+}{% adds the current row style to the next column
  >{\leavevmode\@rowstyle\ignorespaces}%
}

\IEEEoverridecommandlockouts

\begin{document}
\title{Lifting Private Information Retrieval from Two to any Number of Messages}

\author{Rafael G.L. D'Oliveira, Salim El Rouayheb \thanks{This work   was supported in part by NSF Grant CCF 1817635.} \\ ECE, Rutgers University, Piscataway, NJ\\ Emails: rd746@scarletmail.rutgers.edu, salim.elrouayheb@rutgers.edu }

%\author{% aefraer \IEEEauthorrefmark{baer}
%\IEEEauthorblockN{Rawad Bitar,  Salim El Rouayheb\\ ECE Department, IIT, Chicago\\ rbitar@hawk.iit.edu, salim@iit.edu}
%\IEEEauthorblockN{Rafael G.L. D'Oliveira and Salim El Rouayheb
%\thanks{R.G.L. D'Oliveira and S. El Rouayheb are with the ECE department at Rutgers University Emails: rbitar@hawk.iit.edu, parimal@ece.iisc.ernet.in, salim@iit.edu.}
%\thanks{This work was supported in parts by ARL Grant W911NF-17-1-0032.}}
%}

\maketitle

\begin{abstract}
We study private information retrieval (PIR) on coded data with possibly colluding servers. Devising PIR schemes with optimal download rate  in the case of collusion and coded data is still open in general. We provide a lifting operation that can transform  what we call one-shot PIR schemes for two messages into schemes for any number of  messages. We apply this lifting operation on existing PIR schemes and describe two immediate implications. First, we obtain novel PIR schemes with improved download rate in the case of MDS coded data and server collusion. Second, we provide a simplified description of  existing optimal PIR schemes on replicated data as lifted secret sharing based PIR. 

\end{abstract}

\section{Introduction}

We consider the  problem of designing  private information retrieval (PIR) schemes  on coded data stored on multiple servers that can possibly collude. 
In this setting, a user wants to download a message from a server with $M$ messages while revealing no information, in an information-theoretic sense, about which message it is interested in. The database is  replicated on $N$ servers, or in general, could be stored using an erasure code, typically a Maximum Distance Separable (MDS) code\footnote{The assumption here is that messages are divided into chunks which are  encoded separately into $n$ coded chunks using the same code.}. 
 These servers could possibly collude to gain information about the identity of the user's retrieved message. 

The PIR problem was first introduced and studied in \cite{PIR1995, chor1998private} and was followed up by a large body  of work (e.g. \cite{gasarch2004survey,sun2016capacitynoncol, sun2016capacity, yekhanin2010private, beimel2001information, beimel2002breaking}). The model there assumes the database to be replicated and focuses on PIR schemes with efficient total communication rate, i.e., upload and download.  Motivated by big data applications and recent advances in the theory of codes for distributed storage, there has been a growing interest in designing PIR schemes that can query data that is stored in coded form and not just replicated. For this setting, the assumption has been that the messages being retrieved are very large (compared to the queries) and therefore the focus has been on designing PIR schemes that minimize the  download rate.  Despite  significant recent progress, the problem of characterizing the optimal PIR download rate (called PIR capacity) in the case of coded data and server collusion remains open in general.

%PIR schemes for data stored under coded form and not just replicated \cite{shah2014one, chan2014private, tajeddine2016private, extended, banawan2016capacity, fazeli2015pir, blackburn2016pir, freij2016private}. The next example, taken from \cite{tajeddine2016private, extended},  illustrates the construction of a  PIR scheme on coded data.

%
%
%Since its introduction in \cite{PIR1995}, the model of PIR assumes the data to be replicated on multiple nodes (e.g. \cite{sun2016capacitynoncol, sun2016capacity, yekhanin2010private, beimel2001information, beimel2002breaking}). 
%Recently, there has been a growing interest in using codes in DSS to minimize the storage overhead of data. This has motivated recent works on PIR schemes for data stored under coded form and not just replicated \cite{shah2014one, chan2014private, tajeddine2016private, extended, banawan2016capacity, fazeli2015pir, blackburn2016pir, freij2016private}. The next example, taken from \cite{tajeddine2016private, extended},  illustrates the construction of a  PIR scheme on coded data.
%%in \cite{banawan2016capacity}. 

\noindent{\em{Related work:}} 
%Until recently, most of the work on PIR has focused on replicated data and minimizing the total download cost \cite{beimel2001information, beimel2002breaking,dvir20142, yekhanin2008towards,yekhanin2010private, sun2016capacitynoncol, sun2016capacity, efremenko20123}. 
%PIR schemes for replicated data: When PIR was first studied, the model was as follows: a server of $m$ files is replicated on $n$ servers, and a user wants to retrieve a record from this server without revealing the identity of the record he wants, and also minimizing the total communication cost (upload and download). Significant work has been done on the topic \cite{beimel2000reducing, beimel2005general, beimel2002breaking,dvir20142,beimel2001information, yekhanin2008towards,yekhanin2010private, , sun2016capacitynoncol}. Also, work has been done when some of the nodes are colluding and cooperating to find out what the user wants \cite{sun2016capacity}. 
When the data is replicated,  the problem of finding the PIR capacity,  i.e., minimum download rate, is essentially solved.   It was shown in \cite{sun2016capacitynoncol} and \cite{sun2016capacity} that the  PIR capacity  is $(1+T/N+T^2/N^2+\dots+T^{M-1}/N^{M-1})^{-1}$, where $N$ is the number of servers, $T$ is the number of colluding servers and $M$ is the number of messages. Capacity achieving PIR schemes were also presented in \cite{sun2016capacitynoncol} and \cite{sun2016capacity}. 

 When the data is coded  and stored on a large number of servers (exponential in the number of messages), it was shown  in \cite{shah2014one} that downloading one extra bit is enough to achieve privacy. In \cite{chan2014private}, the authors derived bounds on the tradeoff between storage cost and download cost for linear coded data and studied properties of PIR schemes on MDS data. Explicit constructions of efficient PIR scheme on MDS data were first presented in  \cite{tajeddine2016private} for both collusions and no collusions. Improved PIR schemes for MDS coded data with collusions were presented in \cite{freij2016private}. PIR schemes for general linear codes, not necessarily MDS, were studied in \cite{kumar2017private}. The PIR capacity for MDS coded data and no collusion was determined in  \cite{banawan2016capacity}, and remains unknown for the case of collusions.  

\noindent{\em Contributions:}  We introduce what we refer to as a  lifting operation that transforms a class of \emph{one-shot} linear PIR schemes that can retrieve privately one out of a total of two messages, into general PIR schemes on any number of messages. 
%The majority of the existing PIR schemes in the literature are one-shot schemes.

In the literature, the majority of  PIR schemes on coded data, such as those in \cite{tajeddine2016private} and \cite{freij2016private}, are one-shot schemes. First, we describe a refinement operation on these schemes that improves their rate for two messages. Then, we describe how the refined version can be lifted to any number of messages. Finally, we apply the lifting operation on existing PIR schemes and describe two immediate implications: 

\begin{itemize}
\item Applying the lifting operation to the schemes presented in\cite{tajeddine2016private} and \cite{freij2016private}, we obtain novel PIR schemes  with improved download rate for MDS coded data and server collusion.
\item The capacity achieving  PIR schemes on replicated data in \cite{sun2016capacitynoncol} and \cite{sun2016capacity} can be seen as lifted secret sharing.
\end{itemize}

%In this paper, we present a construction of universal $\nu$-robust PIR schemes on $(n,k)$ MDS coded data, where $\nu$ is  the maximum number of unresponsive nodes\footnote {The parameter $\nu$ can be between $0$ and $n-k-1$. A $0$-robust scheme is a non-robust scheme.  If $\nu=n-k$, i.e., there is no redundant data queried, then perfect privacy can not be achieved except by downloading all the files.  If $\nu>n-k$, the file can not be fully retrieved since the MDS code cannot tolerate more than $n-k$ failures.}. We focus on non-colluding nodes (i.e., no spy nodes in the model in {\blue \cite{tajeddine2016private,extended}}) and want to achieve perfect privacy which guarantees that zero information is leaked to the individual nodes about the index of the retrieved file. The construction is a generalization of our PIR schemes on MDS codes in \cite{tajeddine2016private}, with robustness against up to $\nu$ unresponsive nodes. The proposed scheme consists of two layers and has the following properties: (i) {\em universality}, meaning the scheme allows the user to retrieve the requested file privately, for all number of unresponsive servers up to $\nu${\blue , and achieving the} optimal {\blue $cPoP = \frac{n-i}{n-i-k}$ for all $i=1,\dots, \nu$, where $i$ is the actual number of unresponsive nodes}; and (ii) {\em adaptivity}, meaning the scheme changes depending on which nodes do not respond.

\section{Setting}

A set of $M$ messages, $\{ \bm{W}^1, \bm{W}^2, \ldots, \bm{W}^M \} \subseteq \mathbb{F}_q^L$, are stored on $N$ servers each using an $(N,K)$-MDS code. We denote by $ \bm{W}_i^j \in \mathbb{F}_q^{L/K}$, the data about $\bm{W}^j$ stored on server~$i$.
\[
\begin{tabular}{c | c  c  c  c}
 & server $1$ & server $2$ &$\cdots$ & server $N$ \\
\hline
$\bm{W}^1$ & $\bm{W}_1^1$ & $\bm{W}_2^1$ &$\cdots$ & $\bm{W}_N^1$ \\
$\bm{W}^2$ & $\bm{W}_1^2$ & $\bm{W}_2^2$ &$\cdots$ & $\bm{W}_N^2$  \\
$\vdots$ & $\vdots$ & $\vdots$ & $\vdots$ & $\vdots$  \\
$\bm{W}^M$ & $\bm{W}_1^M$ & $\bm{W}_2^M$ &$\cdots$ & $\bm{W}_N^M$ \\
\end{tabular}
\]
Since the code is MDS, each $\bm{W}^j$ is determined by any $K$-subset of $\{ \bm{W}^j_1 , \ldots, \bm{W}^j_N \}$.

The data on server $i$ is $\bm{D}_i = (\bm{W}^1_i , \ldots, \bm{W}^M_i) \in \mathbb{F}_q^{ML/K}$.

A \emph{linear query} (from now on we omit the term linear) is a vector $\bm{q} \in \mathbb{F}_q^{ML/K}$. When a user sends a query $\bm{q}$ to a server $i$, this server answers back with the inner product $\langle \bm{D}_i , \bm{q} \rangle \in \mathbb{F}_q$.

The problem of private information retrieval can be stated informally as follows: A user wishes to download a file $\bm{W}^m$ without leaking any information about $m$ to any of the servers where at most $T$ of them may collude. The goal is for the user to achieve this while minimizing the download rate.

The messages $\bm{W}^1, \bm{W}^2, \ldots, \bm{W}^M$ are assumed to be independent and uniformly distributed elements of $\mathbb{F}_q^L$. The user is interested in a message $\bm{W}^m$. The index of this message, $m$, is chosen uniformly at random from the set $\{1,2,\ldots,M\}$.

A \emph{PIR scheme} is a set of queries for each possible desired message $\bm{W}^m$. We denote a scheme by $\mathcal{Q} = \{ \mathcal{Q}^1, \ldots, \mathcal{Q}^M \}$ where $\mathcal{Q}^m = \{ Q_1^m, \ldots, Q_N^m \}$ is the set of queries which the user will send to each server when they wish to retrieve $\bm{W}^m$. So, if the user is interested in $\bm{W}^m$, $Q^m_i$ denotes the set of queries sent to server $i$. The set of answers,  $\mathcal{A} = \{ \mathcal{A}^1, \ldots, \mathcal{A}^M \}$, is defined analogously.

A PIR scheme should satisfy two properties:
\begin{enumerate}
\item Correctness: $H(\bm{W}^m | \mathcal{A}^m) = 0$.
\item $T$-Privacy:  $I(\cup_{j \in J} \mathcal{Q}_j^m ; m) = 0$, for every $J \subseteq [M]$ such that $|J|=T$, where $[M]=\{1, \ldots, M \}$.
\end{enumerate}

Correctness guarantees that the user will be able to retrieve the message of interest. $T$-Privacy guarantees that no $T$ colluding servers will gain any information on the message in which the user is interested.

\begin{definition}
Let $M$ messages be stored using an $(N,K)$-MDS code on $N$ servers. An $(N,K,T,M)$-PIR scheme is a scheme which satisfies correctness and $T$-Privacy.
\end{definition}

Note that $T$-Privacy implies in  $|\mathcal{Q}^1|=|\mathcal{Q}^i|$ for every $i$, i.e., the number of queries does not depend on the desired message.

\begin{definition}
The \emph{PIR rate} of an $(N,K,T,M)$-PIR scheme $\mathcal{Q}$ is $R_\mathcal{Q}= \frac{L}{|\mathcal{Q}^1|}$.
\end{definition}

%Note that the PIR rate is the reciprocal of download cost.

\section{One-Shot Schemes}

In this section, we introduce the notion of a one-shot scheme, which captures  the majority of the schemes in the literature.

Without loss of generality, we assume that the user is interested in retrieving the first message.  We denote by $$\mathbb{V}_1 = \{ \bm{a} \in \mathbb{F}_q^{ML/K}  : i>L/K \Rightarrow \bm{a}_i = 0 \},$$ the subspace of queries which only query the first message.

\begin{definition}
An $(N,K,T,M)$-one-shot PIR scheme of co-dimension $r$ is an $(N,K,T,M)$-PIR scheme where each server is queried exactly once and in the following way.

\begin{table}[H]
\centering
\begin{tabular}{c c c c c c}
Server $1$ & $\cdots$ & Server $r$ & Server $r+1$ & $\cdots$ & Server $N$ \\
\hline
$\bm{q}_1$ & $\cdots$ & $\bm{q}_r$ & $\bm{q}_{r+1}+\bm{a}_1 $ & $\cdots$ & $\bm{q}_N+\bm{a}_{N-r}$  \\
\end{tabular}
\caption{Query structure for a one-shot scheme.}
\label{table:oneshot}
\end{table}

The queries in Table \ref{table:oneshot} satisfy the following properties:
\begin{enumerate}
\item Any collection of $T$ queries from $\bm{q}_1, \ldots, \bm{q}_N \in \mathbb{F}_q^{ML/K}$ is uniformly and independently distributed. \label{property1}
\item The $\bm{a}_1, \ldots, \bm{a}_{N-r} \in \mathbb{V}_1$ are such that the responses $\langle \bm{D}_{r+1} , \bm{a}_1 \rangle , \ldots, \langle \bm{D}_N , \bm{a}_{N-r} \rangle$ are linearly independent. \label{property2}
\item For $i>r$, the response $\langle \bm{D}_i , \bm{q_i} \rangle$ is a linear combination of $\langle \bm{D}_1 , \bm{q}_1 \rangle , \ldots, \langle \bm{D}_r , \bm{q_r} \rangle$. \label{property3}
\end{enumerate}
\end{definition}

Property \ref{property1} ensures privacy. Properties \ref{property2} and \ref{property3} ensure correctness.

%We now find the rate of a one-shot scheme.

\begin{proposition} \label{pro:oneshotpir}
Let $\mathcal{Q}$ be an $(N,K,T,M)$-one-shot scheme of co-dimension $r$. Then, its rate is given by
\[R_\mathcal{Q} = \frac{N-r}{N} = 1 - \frac{r}{N} .\]
\end{proposition}
\begin{proof}
Since for every $i>r$, $\langle \bm{D}_i , \bm{q_i} \rangle$ is a linear combination of $\langle \bm{D}_1 , \bm{q}_1 \rangle , \ldots, \langle \bm{D}_r , \bm{q_r} \rangle$, the user can retrieve the linearly independent $\langle \bm{D}_{r+1} , \bm{a}_1 \rangle , \ldots, \langle \bm{D}_N , \bm{a}_{N-r} \rangle$.
\end{proof}

Technically,  $N-r$ must be divisible by $L$. When this does not occur, the one-shot scheme must be repeated  $\lcm(N-r,L)$ times\footnote{We denote the least common multiple of $N-r$ and $K$ by $\lcm(N-r,L)$.}. This, however, does not change the rate of the scheme.

We present an example of a one-shot scheme from \cite{tajeddine2016private}.

\begin{example} \label{exe:oneshot}
Suppose the messages are stored using a $(4,2)$-MDS code over $\mathbb{F}_3$ in the following way: 
\[
\resizebox{.95\hsize}{!}{
\begin{tabular}{c | c  c  c  c}
 & Server $1$ & Server $2$ & Server $3$ & Server $4$ \\
\hline
$\bm{W}^1$ & $\bm{W}_1^1$ & $\bm{W}_2^1$ &$\bm{W}_1^1+\bm{W}_2^1$ & $\bm{W}_1^1+2 \bm{W}_2^1$ \\
$\bm{W}^2$ & $\bm{W}_1^2$ & $\bm{W}^2_2$ & $\bm{W}_1^2 + \bm{W}^2_2$ & $\bm{W}_1^2 + 2 \bm{W}^2_2$ \\
$\vdots$ & $\vdots$ & $\vdots$ & $\vdots$ & $\vdots$  \\
$\bm{W}^M$ & $\bm{W}_1^M$ & $\bm{W}^M_2$ & $\bm{W}_1^M + \bm{W}^M_2$ & $\bm{W}_1^M + 2 \bm{W}^M_2$
\end{tabular}
} 
\]
Suppose the user is interested in the first message and wants $2$-privacy, i.e., at most $2$ servers can collude. The following is a $(4,2,2,M)$-one-shot scheme taken from \cite{tajeddine2016private}.
\begin{table}[H]
\centering
\begin{tabular}{c | c c c c}
 & Server $1$ &  Server $2$ & Server $3$ & Server $4$ \\
\hline
Queries & $\bm{q}_1$ & $\bm{q}_2$ & $\bm{q}_3$ & $\bm{q}_4+\bm{e}_1$   \\
Responses & $\langle \bm{D}_1 , \bm{q}_1 \rangle$ & $\langle \bm{D}_2 , \bm{q}_2 \rangle$ & $\langle \bm{D}_3 , \bm{q}_3 \rangle$ & $\langle \bm{D}_4 , \bm{q}_4 \rangle$ 
\end{tabular}
\caption{Query and response structure for Example \ref{exe:oneshot}.}
\label{table:example1}
\end{table}
The queries in Table \ref{table:example1} satisfy the following properties:
\begin{itemize}
\item The queries $\bm{q}_1,\bm{q}_2 \in \mathbb{F}_q^{ML/K}$ are uniformly and independently distributed.
\item We have $\bm{q}_3=\bm{q}_1+\bm{q}_2$ and $\bm{q}_4=\bm{q}_1+2\bm{q}_2$.
\item The query $\bm{e}_1 \in \mathbb{V}_1$ corresponds to the queries (in this case there is only $r=1$ query) $\bm{a}_1, \ldots, \bm{a}_{N-r}$ in \mbox{Table \ref{table:oneshot}}, and is the first vector of the standard basis of $\mathbb{F}_q^{ML/K}$, i.e, $\bm{e}_1$ only has entry $1$ in the first coordinate and $0$ on all the other coordinates.
\end{itemize}

This scheme is private since for any two servers the queries are uniformly and independently distributed.

To retrieve $\langle \bm{D}_4 , \bm{e}_1 \rangle$ the user uses the following identity:
\begin{align} \label{eq:oneshotidentity}
\langle \bm{D}_4 , \bm{q}_4 \rangle = - \langle \bm{D}_1 , \bm{q}_1 \rangle +2 \langle \bm{D}_2 , \bm{q}_2 \rangle + 2 \langle \bm{D}_3 , \bm{q}_3 \rangle .
\end{align}

With this we have one linear combination of $\bm{W}^1$, the first coordinate of $\bm{W}^1_1+\bm{W}^1_2$. Repeating this $\lcm(N-r,L)$ times, we obtain enough combinations to decode $\bm{W}^1$.

To retrieve $1$ unit of the message the user has to download $4$ units. Therefore, the rate of the PIR scheme is $R=1/4$, which could have also been obtained from Proposition \ref{pro:oneshotpir}.
\end{example}

\section{The Refinement Lemma}

The rate of a one-shot scheme is independent of the number of messages. In this section, we show how to refine a one-shot scheme to obtain a better rate for the case of two messages.

Analogous to $\mathbb{V}_1$, we denote by
\[ \mathbb{V}_2 = \{ \bm{b} \in \mathbb{F}_q^{ML/K}  : \text{$i < L/K+1$ or $i>2L/K$} \Rightarrow \bm{b}_i = 0 \}\]
the subspace of queries which only query the second message.

\begin{lemma}[The Refinement Lemma] \label{lem:ref}
Let $\mathcal{Q}$ be a one-shot scheme of co-dimension $r$, with rate $\frac{N-r}{N}$. Then, there exists  an $(N,K,T,2)$-PIR scheme, $\mathcal{Q'}$, with rate $R_{\mathcal{Q'}} = \frac{N}{N+r} > \frac{N-r}{N}$.
\end{lemma}
\begin{proof}
We construct $\mathcal{Q}'$ in the following way.
\begin{table}[H]
\centering
\begin{tabular}{c c c c c c}
Server $1$ & $\cdots$ & Server $r$ & Server $r+1$ & $\cdots$ & Server $N$ \\
\hline
 $\bm{a}_1$ & $\cdots$ & $\bm{a_r}$ & $\bm{a_{r+1}}+\bm{b_{r+1}} $ & $\cdots$ & $\bm{a}_N+\bm{b_}N$  \\
 $\bm{b}_1$ & $\cdots$ & $\bm{b_r}$ & & &
\end{tabular}
\caption{Query structure for $\mathcal{Q}'$.}
\label{table:refinementlemma}
\end{table}
The queries in Table \ref{table:refinementlemma} satisfy the following properties:
\begin{itemize}
\item The queries $\bm{a}_i \in \mathbb{V}_1$ and $\bm{b}_i \in \mathbb{V}_2$.
\item Each query $\bm{b}_i$ is chosen with distribution induced by the query $\bm{q}_i$ of the one-shot scheme $\mathcal{Q}$.\footnote{A probability distribution on $\mathbb{F}_q^{ML/K}$ induces a probability distribution on $\mathbb{V}_2 \subseteq \mathbb{F}_q^{ML/K}$. }
\item Any subset of size $T$ of $\langle \bm{D}_{1} , \bm{b}_1 \rangle , \ldots, \langle \bm{D}_N , \bm{b}_N \rangle$ is linearly independent\footnote{For large fields this occurs with high probability.\label{footnote}}.
\item For $i \leq r$, $\bm{a_i}$ is chosen with distribution identical to $\bm{b}_i$.
\item For $i>r$, $\bm{a_i}$ is chosen such that the set of responses $\langle \bm{D}_{1} , \bm{a}_1 \rangle , \ldots, \langle \bm{D}_N , \bm{a_{N}} \rangle$ is linearly independent$^{\ref{footnote}}$.
\end{itemize}

Privacy is inherited from the one-shot scheme by randomizing the order in which the queries to a server are sent.

The following is also inherited from the one-shot scheme: For $i>r$, $\langle \bm{D}_i , \bm{b_i} \rangle$ is a linear combination of $\langle \bm{D}_1 , \bm{b}_1 \rangle , \ldots, \langle \bm{D}_r , \bm{b_r} \rangle$. Thus, the user can retrieve the linearly independent  $\langle \bm{D}_{1} , \bm{a}_1 \rangle , \ldots, \langle \bm{D}_N , \bm{a_{N}} \rangle$.
\end{proof}

We now apply the refinement lemma to Example~\ref{exe:oneshot}.

\begin{example} \label{exe:refined}
Consider Example~\ref{exe:oneshot} but with two messages. Applying the refinement lemma, we get the following scheme.
\begin{table}[H]
\centering
\begin{tabular}{c c c c}
 Server $1$ &  Server $2$ & Server $3$ & Server $4$ \\
\hline
$\bm{a}_1$ & $\bm{a}_2$ & $\bm{a}_3$ &  $\bm{a}_4+ \bm{b}_4$   \\
  $\bm{b}_1$ &  $\bm{b}_2$ &  $\bm{b}_3$ &   
\end{tabular}
\caption{Query structure of the refinement of Table \ref{table:example1}.}
\label{table:example2}
\end{table}
The queries in Table \ref{table:example2} satisfy the following properties:
\begin{itemize}
\item The queries $\bm{a}_i \in \mathbb{V}_1$ and $\bm{b}_i \in \mathbb{V}_2$.
\item The queries $\bm{b}_1$ and $\bm{b}_2$ are uniformly and independently distributed and are linearly independent.
\item We have $\bm{b}_3=\bm{b}_1+\bm{b}_2$ and $\bm{b}_4=\bm{b}_1+2\bm{b}_2$.
\item The queries $\bm{a}_1$ and $\bm{a}_2$ are uniformly and independently distributed and are linearly independent.
\item We have $\bm{a}_3=\bm{a}_1+\bm{a}_2$ and $\bm{a}_4=\bm{a}_1+3\bm{a}_2$.
\end{itemize}

Privacy is inherited from the one-shot scheme.

To retrieve $\langle \bm{D}_4 , \bm{a}_4 \rangle$ we use the following identity (inherited from \eqref{eq:oneshotidentity} in the one-shot scheme): 
\[ \langle \bm{D}_4 , \bm{b}_4 \rangle = - \langle \bm{D}_1 , \bm{b}_1 \rangle +2 \langle \bm{D}_2 , \bm{b}_2 \rangle + 2 \langle \bm{D}_3 , \bm{b}_3 \rangle .\]

The choice $\bm{a}_4=\bm{a}_1+3\bm{a}_2$ is done so that the set of responses $\langle \bm{D}_{1} , \bm{a}_1 \rangle , \ldots , \langle \bm{D}_4 , \bm{a_{4}} \rangle$ is linearly independent.

As per Lemma \ref{lem:ref}, the rate of this scheme is $R=4/7$, larger than the rate of $1/4$ in Example \ref{exe:oneshot}.

\end{example}

\begin{remark}
The scheme in Example \ref{exe:oneshot} is defined over the field $\mathbb{F}_3$. However, the refinement of this scheme in Example~\ref{exe:refined} requires a larger field since we need the coefficient $3$ in \mbox{$a_4 = a_1 + 3a_2$} so that the set $\{ \langle \bm{D}_{1} , \bm{a}_1 \rangle , \ldots , \langle \bm{D}_4 , \bm{a_{4}} \rangle \}$ is linearly independent. In our schemes, we will assume that the base field is large enough.
\end{remark}

\section{The Lifting Theorem}

In this section, we present our main result in Theorem \ref{teo:lift}. We show how to extend, by means of a lifting operation, the refined scheme on two messages to any number of messages. Informally, the lifting operation consists of two steps: a symmetrization step, and a way of dealing with ``leftover'' queries that result from the symmetrization. 
We also introduce a symbolic matrix representation for PIR schemes which simplifies our analysis. 

\subsection{An Example of the Lifting Operation} \label{sec:exlift}

We denote by 
\begin{equation*}
\resizebox{.95\hsize}{!}{$\mathbb{V}_j = \{ \bm{b} \in \mathbb{F}_q^{ML/K}  : \text{$i < (j-1)L/K+1$ or $i>jL/K$} \Rightarrow \bm{b}_i = 0 \}$,}
\end{equation*}
the subspace of queries which only query the $j$-th message.

\begin{definition}
A $k$-query is a sum of $k$ queries, each belonging to a different $\mathbb{V}_j$, $j \in [M]$.
\end{definition}

So, for example, if $\bm{a}\in \mathbb{V}_1$, $\bm{b}\in \mathbb{V}_2$, and $\bm{c}\in \mathbb{V}_3$, then $\bm{a}$ is a $1$-query, $\bm{a}+\bm{b}$ is a $2$-query, and $\bm{a}+\bm{b}+\bm{c}$ is a $3$-query.

Consider the scheme in Example \ref{exe:refined}. We represent the structure of this scheme by means of the following matrix:
\begin{align} \label{eq:s2}
 S_2 = \begin{pmatrix}
1 & 1 & 1 & 2
\end{pmatrix}
.\end{align}

Each column of $S_2$ corresponds to a server. A $1$ in column~ $i$ represents sending all possible combinations of $1$-queries of every message to server $i$, and a $2$ represents sending all combinations of $2$-queries of every message to server $i$. We call this matrix the \emph{symbolic matrix} of the scheme.

The co-dimension $r=3$ tells us that for every $r=3$ ones there is $N-r=1$ twos in the symbolic matrix.

Given the interpretation above, the symbolic matrix $S_2$ can be readily applied to obtain the structure of a PIR scheme for any number of messages $M$. For $M=3$, the structure is as follows.
\begin{table}[H]
\centering
\begin{tabular}{c  c  c  c}
server $1$ & server $2$ & server $3$ & server $4$ \\
\hline
$\bm{a}_1$ & $\bm{a}_2$ & $\bm{a}_3$ & $\bm{a}_4 + \bm{b}_4 $  \\
$\bm{b}_1$ & $\bm{b}_2$ & $\bm{b}_3$ &  $\bm{a}_5 + \bm{c}_4 $\\
$\bm{c}_1$ & $\bm{c}_2$ & $\bm{c}_3$ &  $\bm{b}_5 + \bm{c}_5$\\
\end{tabular}
\caption{Query structure for $M=3$ in Example \ref{exe:oneshot} as implied by the symbolic matrix in \eqref{eq:s2}.}
\label{table:partiallift}
\end{table}

The relationships between the queries in Table \ref{table:partiallift} is taken from the one-shot scheme and satisfy the following properties:

\begin{itemize}
\item The $a_i \in \mathbb{V}_1$, $b_i \in \mathbb{V}_2$, and $c_i \in \mathbb{V}_3$.
\item The $a$'s and $b$'s are chosen as in Example \ref{exe:refined}.
\item The $c$'s are chosen analogously to the $b$'s.
\item The extra ``leftover" term $\bm{b}_5+\bm{c}_5$ is chosen uniformly and independent and different from zero.
\end{itemize}

The scheme in Table \ref{table:partiallift} has rate $5/12$. In this scheme, the role of $\bm{b}_5+\bm{c}_5$ is to achieve privacy and does not contribute to the decoding process. In this sense, it can be seen as a ``leftover'' query of the symmetrization. By repeating the scheme $r=3$ times, each one shifted to the left, so that the ``leftover'' queries appear in different servers, we can apply the same idea in the one-shot scheme to the ``leftover" queries, as shown in Table~\ref{table:lifted}. Thus, we   improve the rate from $5/12$ to $16/37$.

\begin{table}[H]
\centering
\begin{tabular}{c  c  c  c}
server $1$ & server $2$ & server $3$ & server $4$ \\
\hline
$\bm{a}_1$ & $\bm{a}_2$ & $\bm{a}_3$ & $\bm{a}_4 + \bm{b}_4 $  \\
$\bm{b}_1$ & $\bm{b}_2$ & $\bm{b}_3$ &  $\bm{a}_5 + \bm{c}_4 $\\
$\bm{c}_1$ & $\bm{c}_2$ & $\bm{c}_3$ &  $\bm{b}_5 + \bm{c}_5$\\ \hdashline[2pt/1pt]
$\bm{a}_7$ & $\bm{a}_8$ & $\bm{a}_9 + \bm{b}_9 $  & $\bm{a}_6$  \\
$\bm{b}_7$ & $\bm{b}_8$ & $\bm{a}_{10} + \bm{c}_9 $  &  $\bm{b}_6$ \\
$\bm{c}_7$ & $\bm{c}_8$ & $\bm{b}_{10} + \bm{c}_{10}$  & $\bm{c}_6$ \\
\hdashline[2pt/1pt]
$\bm{a}_{13}$ & $\bm{a}_{14} + \bm{b}_{14} $ & $\bm{a}_{11}$ & $\bm{a}_{12}$   \\
$\bm{b}_{13}$ & $\bm{a}_{15} + \bm{c}_{14} $ & $\bm{b}_{11}$ & $\bm{b}_{12}$ \\
$\bm{c}_{13}$ & $\bm{b}_{15} + \bm{c}_{15}$ & $\bm{c}_{11}$ & $\bm{c}_{12}$ \\
\hdashline[2pt/1pt]
$\bm{a}_{16}+\bm{b}_{16}+\bm{c}_{16}$ &  &  &   \\
\end{tabular}
\caption{Query structure for the lifted scheme.}
\label{table:lifted}
\end{table}

The queries in Table \ref{table:lifted} satisfy the following properties:

\begin{itemize}
\item The scheme is separated into four rounds.
\item In each of the first three rounds the queries behave as in Table \ref{table:partiallift}, but shifted to the left so that the ``leftover'' queries appear in different servers.
\item But now, $\bm{b}_{16}+\bm{c}_{16}$, $\bm{b}_{15}+\bm{c}_{15}$, $\bm{b}_{10}+\bm{c}_{10}$, and $\bm{b_{5}}+\bm{c_{5}}$ are chosen analogously to the one-shot scheme.
\end{itemize}

More precisely, $\bm{b}_{16}+\bm{c}_{16}$ and $\bm{b}_{15}+\bm{c}_{15}$ are uniformly and independently distributed and are linearly independent, \[\bm{b}_{10}+\bm{c}_{10}= (\bm{b}_{16}+\bm{c}_{16}) + (\bm{b}_{15}+\bm{c}_{15}) \] and \[ \bm{b_{5}}+\bm{c_{5}} = (\bm{b}_{16}+\bm{c}_{16} )+ 2 (\bm{b}_{15}+\bm{c}_{15}) .\]

In this way, $\langle \bm{D}_1, \bm{a}_{16} \rangle$ can be retrieved using the following identity (analogous to Example \ref{exe:refined}):
\begin{multline*}\
\langle \bm{D}_1, \bm{b}_{16}+\bm{c}_{16} \rangle = 2 \langle \bm{D}_2, \bm{b}_{15}+\bm{c}_{15} \rangle  + 2 \langle \bm{D}_3, \bm{b}_{10}+\bm{c}_{10} \rangle \\ - \langle \bm{D}_1, \bm{b}_{5}+\bm{c}_{5} \rangle 
\end{multline*}

This scheme can be represented by the following matrix\footnote{We omit zeros in our symbolic matrices. }.
\[ S_3 = \begin{pmatrix}
1 & 1 & 1 & 2 \\
1 & 1 & 2 & 1 \\
1 & 2 & 1 & 1 \\
3 & & &
\end{pmatrix} \]

The scheme for $M=3$ messages was constructed recursively using the one for $2$ messages. It is this recursive operation that we call lifting. The main idea behind the lifting operation is that $r=3$ entries with value $k$ generate $N-r=1$ entry with value $k+1$ in the symbolic matrix.

Lifting $S_3$ to $S_4$ follows the same procedure: repeat $S_3$ $r=3$ times, each one shifted to the left, to produce $N-r=1$ $4$-query. As a result, we obtain the following symbolic matrix.

\[ S_4 = \begin{pmatrix}
1 & 1 & 1 & 2 \\
1 & 1 & 2 & 1 \\
1 & 2 & 1 & 1 \\
3 & & & \\
1 & 1 & 2 & 1 \\
1 & 2 & 1 & 1 \\
2 & 1 & 1 & 1 \\
 & & & 3 \\
1& 2 & 1 & 1 \\
2 & 1& 1 & 1 \\
1 & 1 & 1 & 2 \\
 & & 3 & \\
  & 4 &  & 
\end{pmatrix} \]

The queries are to be chosen analogously to the previous examples which we describe rigorously in the next subsection.

\subsection{The Symbolic Matrix and the Lifting Operation}

\begin{definition}
Let $\mathcal{Q}$ be a one-shot scheme with co-dimension $r$. A symbolic matrix $S_M$ for $\mathcal{Q}$ is defined recursively as follows.
\begin{align} 
S_2 &= (\overbrace{1,\ldots,1}^{r}, \overbrace{2,\ldots,2}^{N-r}) \label{eq:S_2} \\
S_{M+1} &= \lift(S_M) \label{eq:symbolicmatrix}
\end{align}

The lifting operation is defined as,
\[ \lift(S_M) = \begin{pmatrix}
S_M \\
\sigma (S_M) \\
\vdots \\
\sigma^{r-1} (S_M)\\
A
\end{pmatrix}, \]
where $\sigma(S_M)$ shifts the columns of $S_M$ to the left and $A$ is a matrix which we will describe later in detail.

Formally, $\sigma(S_M)(i,j) = S_M(i,j+1)$ for $1\leq j \leq n-1$ and $\sigma(S_M)(i,n) = S_M(i,1)$. Here, $\sigma(S_M)(i,j)$ is the entry in the $i$-row and $j$-th column of the matrix $\sigma(S_M)$.

The matrix $A$ is constructed as follows: We first define an ordering\footnote{This is known in the literature as a lexicographical order.}, $\prec$, on $(i,j) \in \mathbb{N}^2$ by $(i,j) \prec (i',j')$ if either $i<i'$ or $i=i'$ and $j<j'$.

Let $B = \{(i,j): S_M (i,j) = M \}= \{b_1, \ldots, b_{\#(M,S_M)} \}$ such that $i<j$ implies in $b_i \prec b_j$, where \[\#(k,A)= | \{ (i,j)\in [n]\times [m] : A_{ij}=k \} | .\]

Define \[\tau(i,j) = \left\{\begin{matrix}
(i+\rows(S_{M}),j-1) & \text{if} \hspace{5pt} j>1, \\ 
(i+\rows(S_{M}),N) & \text{if} \hspace{5pt} j=1,
\end{matrix}\right. \]
where $\rows(S_{M})$ is the number of rows in $S_{M}$.

Define the auxiliary sets $B_i = \{b_i, \tau(b_i), \ldots, \tau^{r-1}(b_i) \}$ and $c(B_i) = \{ j : (i,j) \in B_i)$.

Then, $A$ is defined as \[ A_{i,j} = \left\{\begin{matrix}
0 & \text{if} \hspace{5pt} j\in c(B_i),\\ 
M+1 & \text{if} \hspace{5pt} j\notin c(B_i).
\end{matrix}\right.   \]
\end{definition}

As an example, we show how $S_4$ in Section \ref{sec:exlift} is constructed in terms of $S_3$. In this case, $r=3$ and the matrix $A$ consists of a single row.

\begin{center}
\includegraphics[scale=1]{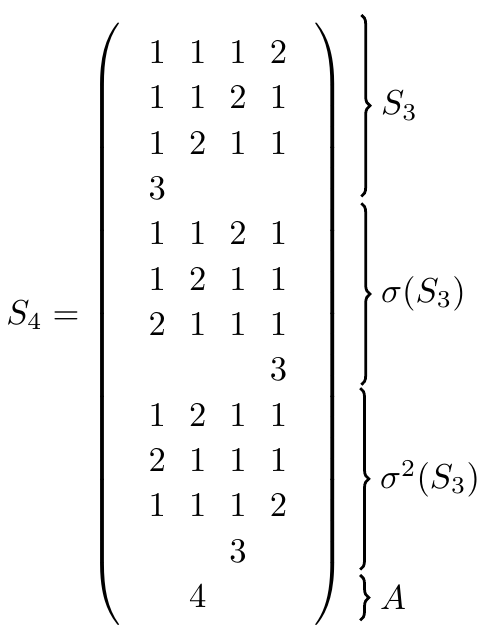}
\end{center}

\noindent{\bf Translation from symbolic matrix to PIR scheme.} 
Each entry $k$ of  a symbolic matrix $S_M$ represents $\binom{M}{k}$ $k$-queries, one for every combination of $k$ messages. The queries are taken analogously to the queries in the one-shot scheme. This is done by making the $r$ queries represented by \mbox{$B_i = \{b_i, \tau(b_i), \ldots, \tau^{r-1}(b_i) \}$} to generate the $N-r$ queries represented by $\{ (i,j) : A_{i,j} = M+1 \}$.

To find the rate of the lifted scheme we need to count the number of entries in the symbolic matrix of a specific value.

\begin{proposition} \label{pro:count}
Let $S_M$ be the symbolic matrix of a one-shot scheme with co-dimension $r$. Then,
\[ \#(k,S_M) =  \left( \frac{N-r}{r} \right)^{k-1} r^{M-1} \hspace{10pt} 1 \leq k \leq M \]
\end{proposition}
\begin{proof}
It follows from the lifting operation that
\begin{align*}
\#(k,S_M) &= \frac{N-r}{r} \#(k-1,S_M) \\
&= \left( \frac{N-r}{r} \right)^{k-1} \#(1,S_M) \\
&= \left( \frac{N-r}{r} \right)^{k-1} r^{M-1}.
\end{align*}
\end{proof}

\begin{theorem} \label{teo:lift}
Let $\mathcal{Q}$ be a one-shot scheme of co-dimension $r$. Then, refining and lifting  $\mathcal{Q}$ gives an $(N,K,T,M)$-PIR scheme $\mathcal{Q}'$ with rate
\[ R_{\mathcal{Q}'} = \frac{(N-r)N^{M-1}}{N^M - r^M} = \frac{N-r}{N \left( 1 - \left(\frac{r}{N}\right)^M\right)} .\]
\end{theorem}

\begin{proof}
Given the one shot-scheme $\mathcal{Q}$, we apply the refinement lemma to obtain a scheme with symbolic matrix $S_2$ as in  \eqref{eq:S_2}. The scheme $\mathcal{Q}'$ is defined as the one with symbolic matrix \mbox{$S_M = \lift^{M-2} (S_2)$} as in \eqref{eq:symbolicmatrix}.\footnote{The power in the expression $\lift^{M-2} (S_2)$ denotes functional composition.} Privacy and correctness of the scheme follow directly from the privacy and correctness of the one-shot scheme.

Next, we calculate the rate $R_{\mathcal{Q}'} = \frac{L}{|{\mathcal{Q}'}^1|}$.
Each entry $k$ of $S_M$ corresponds to $\binom{M}{k}$ $k$-queries, one for each combination of $k$ messages. Thus, using Proposition \ref{pro:count},
\begin{align*}
|{\mathcal{Q}'}^1| &= \sum_{k=1}^M \#(k,S_M) \binom{M}{k} \\
&= r^{M-1} \sum_{k=1}^M \left( \frac{N-r}{r} \right)^{k-1} \binom{M}{k} \\
&= \frac{N^M - r^M}{N-r}
\end{align*}

To find $L$ we need to count the queries which query $\bm{W}_1$. The number of $k$-queries which query $\bm{W}_1$ is $\binom{M-1}{k-1}$. Thus,
\begin{align*}
L &= \sum_{k=1}^M \#(k,S_M) \binom{M-1}{k-1} \\
&= r^{M-1} \sum_{k=1}^M \left( \frac{N-r}{r} \right)^{k-1} \binom{M-1}{k-1} \\
&= N^{M-1}
\end{align*}

Therefore, $R_{\mathcal{Q}'} = \frac{(N-r)N^{M-1}}{N^M - r^M}$.

\end{proof}

\section{Refining and Lifting Known Schemes}

In this section, we refine and lift known one-shot schemes from the literature.

We first refine and lift the scheme described in \mbox{Theorem $3$} of \cite{taje18}. In our notation, this scheme is a one-shot scheme with co-dimension $r=\frac{NK-N+T}{K}$. 

\begin{theorem} \label{teo:salim}
Refining and lifting the scheme presented in \mbox{Theorem $3$} of \cite{taje18} gives an $(N,K,T,M)$-PIR scheme $\mathcal{Q}$ with
\begin{align} \label{eq:salim}
R_{\mathcal{Q}} = \frac{(N+T).(NK)^{M-1}}{(NK)^M - (NK-N+T)^M} .
\end{align} 
\end{theorem}

Next, we refine and lift the scheme in \cite{freij2016private}. In our notation, this scheme has co-dimension $r=K+T-1$. Thus, we obtain the first PIR scheme to achieve the rate conjectured to be optimal\footnote{The optimality of the rate in \eqref{eq:hollanti} was disproven in  \cite{jafar2018conjecture} for some parameters. For the remaining range of parameters, the PIR schemes obtained here in Theorem~\ref{teo:holl}, through refining and lifting, achieve the best rates known so far in the literature.} in \cite{freij2016private} for MDS coded data with collusions. 

\begin{theorem} \label{teo:holl}
Refining and lifting the scheme presented in \cite{freij2016private} gives an $(N,K,T,M)$-PIR scheme, $\mathcal{Q}$, with
\begin{align} \label{eq:hollanti}
R_{\mathcal{Q}} = \frac{(N-K-T+1)N^{M-1}}{N^M - (K+T-1)^M} = \frac{1 - \frac{K+T-1}{N}}{1 - \left( \frac{K+T-1}{N} \right)^M} .
\end{align}
\end{theorem}

The rate of the scheme in Theorem \ref{teo:salim} \eqref{eq:salim} is upper bounded by the rate of the scheme in Theorem \ref{teo:holl} \eqref{eq:hollanti}, with equality when either $K=1$ or $N=K+T$.

Now, we consider the case of replicated data ($K=1$) on $N$ servers with at most $T$ collusions. A $T$-threshold linear secret sharing scheme \cite{CDN15} can be transformed into the following one-shot PIR scheme. \\ \\
\resizebox{!}{0.46cm}{
\begin{tabular}{c  c  c  c c c}
Server $1$ & \ldots & Server $T$ & Server $T+1$ & \ldots & Server $N$ \\
\hline
$\bm{q}_1$ &  & $\bm{q_T}$ & $\bm{q_{T+1}}+\bm{a}_1 $ & & $\bm{q_N + a_{N-T}}$ 
\end{tabular}
} \\ 

\begin{theorem}
Refining and lifting a $T$-threshold linear secret sharing scheme gives an $(N,1,T,M)$-PIR scheme $\mathcal{Q}$ with capacity-achieving rate 
\[ R_{\mathcal{Q}} = \frac{(N-T)N^{M-1}}{N^M - T^M} .\]
\end{theorem}

This scheme has the same capacity achieving rate as the scheme presented in \cite{sun2016capacity} but with less queries.

\bibliographystyle{ieeetr}
\bibliography{coding2,coding1,}

\end{document}